\documentclass[aps,prd,twocolumn,groupedaddress,nofootinbib]{revtex4-2}

\usepackage{amsmath,amssymb,amsthm,graphicx}
\usepackage[dvipsnames]{xcolor}
\usepackage{multirow}
\usepackage{adjustbox}

\newcommand{\Rd}{{\mathord{\mathbb R}^d}}

\newcommand{\id}{{\mathop{\rm \mathbf{id} }}}

\def\P{{\mathcal P}}
\newcommand{\bt}{\mathbf{t}}
\newcommand{\bomega}{\boldsymbol{\omega}}
\newcommand{\bgamma}{\boldsymbol{\gamma}}
\newcommand{\E}{\mathcal{E}}
\newcommand{\tcE}{{\tilde{\mathcal{E}}}}
\newcommand{\tx}{{\tilde{x}}}
\newcommand{\tE}{{\tilde{E}}}
\newcommand{\R}{\mathcal{R}}
\newcommand{\brho}{\boldsymbol{\rho}}
\newcommand{\tbrho}{{\tilde{\boldsymbol{\rho}}}}
\def\P{{\mathcal P}}

\newtheorem{prop}{Proposition}

\newtheorem{cor}{Corollary}

\begin{document}

\title{Linearized Optimal Transport for Collider Events}

\author{Tianji Cai}
\affiliation{Department of Physics, University of California, Santa Barbara, CA 93106, USA}
\author{Junyi Cheng}
\affiliation{Department of Physics, University of California, Santa Barbara, CA 93106, USA}
\author{Katy Craig}
\affiliation{Department of Mathematics, University of California, Santa Barbara, CA 93106, USA}
\author{Nathaniel Craig}
\affiliation{Department of Physics, University of California, Santa Barbara, CA 93106, USA}

\date{\today}

\begin{abstract}
We introduce an efficient framework for computing the distance between collider events using the tools of Linearized Optimal Transport (LOT). This preserves many of the advantages of the recently-introduced Energy Mover's Distance, which quantifies the “work” required to rearrange one event into another, while significantly reducing the computational cost. It also furnishes a Euclidean embedding amenable to simple machine learning algorithms and visualization techniques, which we demonstrate in a variety of jet tagging examples. The LOT approximation lowers the threshold for diverse applications of the theory of optimal transport to collider physics. 
\end{abstract}

\maketitle

\section{Introduction \label{sec:intro}}
What is the distance between collider events? This question, although simple to pose, is notoriously difficult to answer. Identical events at parton level can appear to differ upon reconstruction due to soft or collinear emission, while topologically distinct events at parton level can appear identical upon reconstruction, depending on the degree of coarse-graining.  Despite such challenges, the value of a well-defined distance is clear: the comparison of collider events, or the reconstructed objects contained therein, is an essential step in extracting physics from collider data. 

Significant progress was made towards defining a useful metric on the space of collider events in \cite{Komiske:2019fks}, where the ``Energy Mover's Distance'' (EMD) was introduced to compare the energy flow between events. Properly speaking, the Energy Mover's Distance is an adaptation of the Earth Mover's Distance, itself an example of the $p$-Wasserstein distance appearing in the theory of optimal transport.  Intuitively, the $p$-Wasserstein distance between two normalized energy distributions represents the minimal amount of work required to rearrange one distribution to look like the other, and may be modified (as in the EMD of \cite{Komiske:2019fks}) to accommodate events with different total energies. 

As observed in \cite{Komiske:2019fks} and further developed in \cite{Komiske:2020qhg}, the EMD has numerous applications to collider physics.
Among other things, it provides a new perspective on existing jet variables; implies inequalities satisfied nonperturbatively by jet observables;  and enables the definition of a distance between theories (where theories are defined as collections of events weighted by cross sections). From a practical perspective, the EMD defines new quantities associated with collider events that can be used as input to machine learning (ML) algorithms  and leveraged in collider analyses, providing a novel intermediary between simple analytic variables and deep neural networks. The EMD defined in \cite{Komiske:2019fks} has been subsequently applied to distance-based analysis of jets in CMS Open Data \cite{Komiske:2019jim}, to the definition of a new `event isotropy' shape variable \cite{Cesarotti:2020hwb}, as a metric for variational autoencoder-based anomalous jet tagging \cite{Cheng:2020dal}, and (with suitable generalization) to discrimination at the full event level \cite{Romao:2020ojy}. A number of other metrics for collider events have been explored in \cite{Mullin:2019mmh}. Broadly speaking, the many applications of the EMD pursued in \cite{Komiske:2019fks, Komiske:2019jim, Komiske:2020qhg, Cesarotti:2020hwb, Romao:2020ojy, Cheng:2020dal} highlight the potential relevance of tools from the theory of optimal transport for collider physics. 
 
However, one of the major practical challenges to the use of EMD in analyzing collider events is the computational cost; for a data set containing $N_{\rm evt}$ events, computing the pairwise distance between all events is $\mathcal{O}(N_{\rm evt}^2)$\footnote{The possibility of reducing such classical $\mathcal{O}(N_{\rm evt}^2)$ strategies to $\mathcal{O}(N_{\rm evt})$ quantum algorithms was pointed out in \cite{Wei:2019rqy}.}. This poses a challenge given that computing the $p$-Wasserstein distance between two events itself takes fractions of a second, putting the calculation of EMDs between events in typical collider data sets beyond the reach of desktop computers. It is also unsuitable for use with ML methods that require more structure than just the pairwise distances between events. 
 
In this paper, we define an efficient framework for computing the distance between collider events by applying the tools of Linearized Optimal Transport (LOT), preserving the many advantages of the EMD while significantly reducing the computational cost and furnishing a Euclidean embedding suitable for use in a wide range of ML algorithms. In particular, we implement the LOT approximation of the 2-Wasserstein distance, as   introduced in \cite{wang2013linear}. To the extent that the 2-Wasserstein distance has a pseudo-Riemannian structure (unlike $p$-Wasserstein distances with $p \neq 2$, including the $p=1$ Earth Mover's Distance), the LOT approximation amounts to projecting  onto the 2-Wasserstein tangent plane at a chosen reference event and computing simpler $\ell^2$ distances on that plane. We make this point of view rigorous in the appendix, where we prove that, as the reference event in the LOT approximation is refined, LOT converges to the distance between events on the  tangent plane, which provides a well-defined metric on the space of events. 

The LOT approach vastly speeds up the computation of optimal transport distances between collections of $N_{\rm evt}$ events by requiring the determination of only $\mathcal{O}(N_{\rm evt})$ computationally intensive $p$-Wasserstein distances, followed by $\mathcal{O}(N_{\rm evt}^2)$ computationally efficient $\ell^2$ distances.\footnote{Another pseudo-Riemannian structure, reminiscent of the $2$-Wasserstein metric, has also been used to reduce the computational complexity of multi-particle correlators  \cite{Komiske:2019asc}.}   In practice, replacing the traditional optimal transport computation with this linear version reduces the computational effort of the classification task   from a computer cluster to a single PC. Even with this dramatic reduction in computational time, we still achieve comparable accuracy to previous work using the original Wasserstein distances on the classification task. 

Beyond the significant computational speedup, LOT provides an isometric linear embedding into Euclidean space, suitable for use in a wider range of ML algorithms. We demonstrate its utility as input to ML algorithms tasked with discriminating between samples of boosted jets containing diverse Standard Model (SM) and beyond-Standard Model (BSM) particles. Due to the fact that our ML models lack the expressivity of deep neural networks, they will not, in general, achieve the same levels of accuracy. Instead, our approach offers a much clearer interpretation in terms of the underlying physics, while still achieving very good levels of accuracy. For example, it can provide answers to questions regarding what properties are most important in distinguishing them from each other; see Figure \ref{fig:LDA_sigmas_combined_Wt10000}.

This paper is organized as follows: In Section \ref{sec:lot} we review the $p$-Wasserstein distance and the Linearized Optimal Transport approximation to the $p=2$ distance, framed in terms suitable for application to collider events. We then illustrate features of the LOT approximation in the context of jet tagging in Section \ref{sec:objects}, computing LOT pseudo-distances between various classes of boosted jets using an isotropic (in cylindrical coordinates) distribution as a reference event. The utility of   LOT   as an input to simple machine learning algorithms is highlighted in Section \ref{sec:ml}, where we explore the performance of linear discriminate analysis (LDA), $k$-nearest neighbor (kNN), support vector machine (SVM), and $k$-medoids clustering algorithms in the pairwise classification of boosted QCD, $W$, $t$, Higgs, and BSM jets. The comparable performance of models respectively coupled with the LOT and EMD metrics suggests that the former approximation matches the discriminating power of the latter metric while offering considerable computational speedup. It is also readily amenable to visualization, which we demonstrate in a number of examples. We conclude and enumerate a variety of future directions in Section  \ref{sec:conc}. A proof of the convergence of the LOT approximation to a true metric in the continuum limit is reserved for the Appendix.

\section{Linearized Optimal Transport \label{sec:lot}}
 Let an event $\E$ denote a collection of particles at locations $x_i$ in a rectangular domain $\Omega$, with energies $E_i, \tE_j \geq 0$.\footnote{While the detector on which the collision data is recorded is a cylinder, due to the fact that we will translate jets clustered with unit radius parameter to be centered at the origin, we may neglect the periodic boundary conditions in the azimuthal angle and consider the underlying domain to be a rectangle.} Given two events $\E, \tcE $ with the same total energy, $\sum_i E_i = \sum_j \tE_j$, the theory of  optimal transport provides  various notions  of distance between the two events.   In particular, for $p \geq 1$, the   $p$-Wasserstein distance is given by
\begin{align} \label{Wpdef}
W_p( \E, \tcE) &= \min_{g_{ij} \in \Gamma{(\E,\tcE)} }  \left( \sum_{ij} g_{ij}  \|x_i-\tx_j\|^p \right)^{1/p} , \\
\Gamma {(\E,\tcE)} &= \left\{ g_{ij} : g_{ij} \geq 0, \ \sum_j g_{ij} =  {E}_i, \  \sum_i g_{ij} = \tE_j   \right\} , \nonumber
\end{align}
where $\|x_i-\tilde{x}_j \|$ denotes the angular distance on the underlying space $\Omega$, which we  will often refer to as the \emph{ground metric}.
When $p =1$ or 2, $W_p$ is also known as the Earth Mover's Distance or the Monge-Kantorovich distance, respectively. Up to normalizing the energies $E_i, \tE_j$ by dividing through by the total energy of each event, we may assume without loss of generality that the total energy of all events we consider equals $1$.

One interpretation of the $p$-Wasserstein distance is that it represents the minimal amount of ``effort'' required to rearrange the  distribution of energy in $\E$ to match $\tcE$. In this case,   $g_{ij}$ represents the  amount of energy moved from particle $i$ in event $\E$ to particle $j$ in event $\tcE$, and $ \|x_i-\tilde{x}_j\|^p$ represents the ``cost'' of moving energy between the two locations. With this interpretation, $\Gamma{(\E,\tcE)} $ is the set of possible ways to rearrange $\E$ to look like $\tcE$, known as the set of \emph{transportation plans}: any rearrangement $g_{ij}$ can only move nonnegative amounts of energy, the total amount of energy moved from a fixed particle $i$ in $\E$ to all of the particles in $\tcE$ must   coincide with the original energy $E_i$, and, symmetrically, the total amount of energy moved from all of the particles in $\E$ to any fixed particle $j$ in $\tcE$ must   coincide with $E_j'$.
More generally, there are several methods to extend the Wasserstein distance to events $\E$ and $\tcE$ with different total energies, including the version of the Earth Mover's Distance considered in \cite{Komiske:2019fks}, which is a type of  partial optimal transport distance \cite{hanin1992kantorovich,piccoli2014generalized, piccoli2016properties} created by interpolating between the $1$-Wasserstein distance and the total variation norm.

Over the past twenty years, optimal transport   distances have emerged as  important metrics for image classification tasks \cite{thorpe2017transportation,pele2008linear,pele2009fast,rubner2000earth,wang2010optimal,delon2004midway}. These metrics are unique in that they lift the ground metric on the underlying space to the set of probability distributions  on that space. This is in contrast with more traditional metrics, such as the $\ell^2$ norm. For example, in an image based approach, the $\ell^2$ norm  computes the distance between two events $\E$ and $\tcE$ by, first, binning the particles on a grid with $N$ bins; second, representing the energy at each grid location by vectors $v, \tilde{v} \in \mathbb{R}^{N}$; and, third, computing the distance between $\E$ and $\tcE$ via the standard Euclidean norm,
\begin{align} d_{\ell^2(\mathbb{R}^N)}(\E, \tcE) := \left( \sum_{i=1}^{N} | v_i - \tilde{v}_i |^2 \right)^{1/2} .
\end{align}

Unlike the Wasserstein metric, the $\ell^2$ norm does not respect the geometry of the underlying space. For example, suppose each event consists of a single particle with energy 1, the particles are distance $\|x_1-\tilde{x}_1\|$ apart, and   the grid for the $\ell^2$ norm is fine enough so that the particles fall in different bins. Then,
\begin{align*}
W_p(\E, \tcE) = \|x_1-\tilde{x}\| \quad \text{ and } \quad d_{\ell^2(\mathbb{R}^n)}(\E, \tcE) = \sqrt{2} .
\end{align*}
While the $p$-Wasserstein metrics take into account the particles' locations  on the underlying space, this information is neglected by the classical $\ell^2$ norm. This ability to preserve spatial information provides the $p$-Wasserstein metrics with a natural advantage in image classification tasks.

In spite of these   theoretical benefits of optimal transport metrics, wider adoption in image classification   has been slowed by two obstacles:  computational cost and limited choice of classification algorithms.
In terms of computational efficiency, computing the $p$-Wasserstein distance between two events, with $n$ particles in each event, requires $O(n^3)$ operations via Bertsekas' auction algorithm and  $O(n^2 \log(n))$ operations via entropic regularization and the Sinkhorn algorithm \cite{altschuler2017near, bertsekas1981new, bertsekas1988dual, cuturi2013sinkhorn, peyre2019computational}. This is in contrast to the classical $\ell^2$ norm, which is naively $\mathcal{O}(n)$, when the number of bins is chosen proportional to the number of particles, $n \sim N$. In  image classification tasks, the high cost of the $p$-Wasserstein metrics   is compounded by the fact that  one needs to   compute the pairwise $p$-Wasserstein distances between the entire collection of $N_{\rm evt}$ images, requiring $O(N_{\rm evt}^2)$ computations of the distance. In the particular case of classifying jet events, the number of particles per event is relatively small, $n \approx 10^2$, and it is this latter need to compute pairwise distances between a large number of events, $N_{\rm evt} \approx 10^5$, which is the main computational expense. Furthermore, existing work using classical optimal transport metrics  must also cope with the significant computational demands of storing the matrix of pairwise distances.

The goal of the present work is to overcome the problem of high computational cost and limited choice of  algorithms by using the linearized optimal transport approximation of the 2-Wasserstein distance, originally introduced by Wang, et. al. \cite{wang2013linear} as a method for visualizing variation in sets of images. 
Let $\mathcal{R}$ denote the reference event,  a collection of particles at locations $y_i$ with energies $R_i$. For any event $\mathcal{E}$, let $r_{ij}$ denote an optimal transport plan from $\mathcal{R}$ to $\E$, that is, a minimizer of (\ref{Wpdef}). (Note that there may be more than one optimal transport plan between two given events.) In general, a transport plan $r_{ij}$ may send energy from particle $i$ in the reference measure to many different particles in event $\mathcal{E}$. Consider the average of these locations, weighted by how much energy is sent to each and normalized by the amount of energy starting at particle $i$, 
\begin{align}
z_i := \frac{1}{R_i} \sum_{j} r_{ij} x_j
\end{align}
This provides a map   from an event $\E$ to a vector $z_i$ in $n$-dimensional Euclidean space, $\mathbb{R}^n$, where $n$ is the number of particles in the reference jet.

The LOT approximation of the 2-Wasserstein metric  measures the distance between two events $\E$ and $\tcE$ by considering the Euclidean distances between all pairs $(z_i, \tilde{z}_i)$, weighted by the mass starting at particle $i$,
\begin{align} \label{LOTgdef}
LOT_{r,\tilde r}(\E,\tcE) = \left( \sum_{i} R_i \|z_i -  \tilde{z}_i \|^2 \right)^{1/2} .
\end{align}
Note that this approximation  depends on the choice of transport plans $r_{ij}, \tilde{r}_{ij}$.

In Figure \ref{LOTillustration}, we illustrate the LOT-W2 computation and its relationship to the standard 2-Wasserstein metric (OT-W2). The top row shows two optimal transport plans that rearrange a uniform reference jet of 81 constituent particles (green) into  two sample jets (blue and red), according to the exact 2-Wasserstein metric. Grey lines indicate how energy from particle $y_i$ in the reference jet is sent to particle $x_j$ in sample jet 1 or particle $\tx_j$ in sample jet 2. Note that, as there are multiple optimal ways to perform this rearrangement, the rearrangement is not guaranteed to be symmetric: in the top left figure, compare the fifth particle from the left on the bottom row (which splits mass between both blue particles) to the top row (which sends all mass to the right particle). In the bottom left subplot, we illustrate $\tilde{z}_i  - z_i$, to visualize the difference in how the reference jet is rearranged for jet 1 and jet 2. Predictably, we observe that the main difference is   energy goes further to the right in the case of jet 2. The LOT approximation of the 2-Wasserstein distance is computed by taking the sum of the   lengths of the gray vectors squared, weighted by the energy of the reference measure $R_i = 1/81$, so that $LOT_{r, \tilde{r}}(\E, \tcE) \approx 1.07$. Finally, in the lower right subplot, we illustrate the OT-W2 distance between jet 1 and jet 2, which corresponds to moving half of the energy in the jet 1 a distance $1.5$, so $W_2(\E, \tcE) = \left( 1.5^2/2 \right)^{1/2} \approx 1.06$.

\begin{figure}[h]
\includegraphics[scale=0.11]{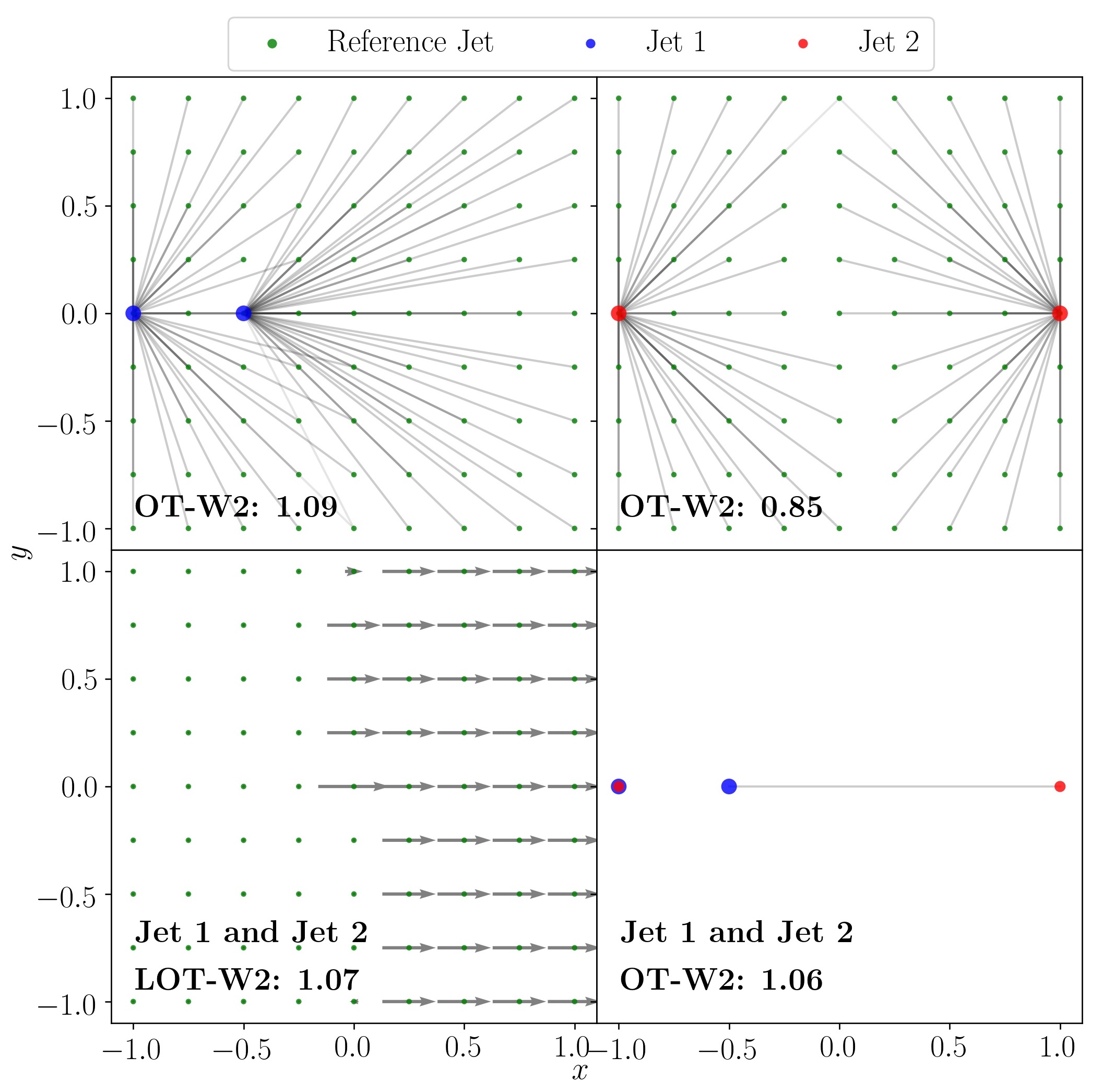}
\caption{\label{artificial_jets} 
\textit{Upper left:} An optimal movement using the OT-W2 metric to rearrange a uniform reference jet of $9 \times 9 =81$ constituent particles (green) into the sample jet 1 (blue).
\textit{Upper right:} An optimal movement using the OT-W2 metric to rearrange the same uniform reference jet (green) into another sample jet 2 (red).
\textit{Lower left:} An optimal movement to rearrange the sample jet 1 into the sample jet 2 using LOT-W2.
\textit{Lower right:} An optimal movement to rearrange the two sample jets directly using OT-W2.} \label{LOTillustration}
\end{figure}

The LOT approximation does not, in general, provide a metric on the space of events. For example, if the reference event $\R$ consists of a single particle at location $y_1$, then ${z}_1 = \sum_{j} x_j E_j$ is the ``center of energy'' of $\E$, and any two events $\E, \tcE$ with equal center of energy   satisfy $LOT_{r, \tilde r} (\E, \tcE) = 0$. Consequently, it is clear that a necessary condition for the LOT approximation to capture finer properties of   events is that the reference event cannot be too concentrated. 
In fact, this condition is also sufficient. In the appendix, we describe how the LOT approximation extends to reference events $\R$ given by general measures on Euclidean space. When the reference event does not concentrate on lower dimensional sets, the LOT approximation coincides with the \emph{transport metric with base $\R$}, denoted $W_{2,\R}$, which is a well-defined metric on the space of events, corresponding to taking the distance between two events by projecting on the 2-Wasserstein tangent plane at $\R$. 
In Corollary \ref{LOTcor} of the appendix, we prove that, if the reference event $\R^{N}$ is given by a collection of $N^2$ particles, uniformly distributed on a rectangle $\Omega$, with equally weighted energies $R_i^N = 1/N^2$, then, as $N \to +\infty$, the LOT approximation  converges to $W_{2, \R}$, where $\R$ is the probability measure uniformly distributed on $\Omega$,
\begin{align}
\lim_{N \to +\infty} LOT_{r^N,{\tilde r}^N}(\E,\tcE) = W_{2, \R} (\E, \tcE) . 
\end{align}
For this choice of $\R$ and any events $\E, \tcE$ on $\Omega$, the  transport metric is bounded above and below by the original 2-Wasserstein distance \cite{merigot2020quantitative},
\begin{align} \label{lottoot}
W_2(\E,\tcE) \leq W_{2, \R}(\E,\tcE) \leq C W_2(\E,\tcE)^{2/15} ,
\end{align}
where the constant $C >0$ depends on   $\Omega$. In this way, LOT not only converges to a well-defined transport metric $W_{2, \R}$, but that transport metric captures the behavior of the original 2-Wasserstein metric at   large and small distances.

The key benefit of the LOT approximation is that it provides a  natural embedding $\E \mapsto z_i$ of events into Euclidean space. This embedding is useful for two reasons. First, $LOT_{r,\tilde r}(\E,\tcE)$  coincides with the $\ell^2$ distance of the Euclidean coordinates $z_i, \tilde{z}_i$, weighted by the energies of the reference measure $ R_i$. Consequently,   to compute the pairwise LOT approximation   between all   events in a sample requires $O(N_{\rm evt})$ computations of the 2-Wasserstein metric, in order to construct the embedding $\E \mapsto z_i$, and then $O(N^2_{\rm evt})$ computations of the $\ell^2$ metric, in order to compute the value of LOT between all events. Given that each computation of $\ell^2$ is, on average, four orders of magnitude faster than computing a Wasserstein distance, this results in an enormous computational advantage. 

The second reason that the LOT Euclidean embedding is useful in jet classification is that it allows us to apply a wider range of classification algorithms directly to the vectors  $z_i, \tilde{z}_i$ representing the events $\E, \tcE$. While existing work using optimal transport for jet classification   considered  algorithms that only rely on pairwise distances between all events, such as kNN, by using the LOT Euclidean embedding, we are  able to apply algorithms that require a Euclidean structure,  such as LDA. By leveraging this Euclidean structure, even this simplistic algorithm is able to provide novel ways to visualize variation in the data set (see Figure \ref{fig:LDA_sigmas_combined_Wt10000}) and surprisingly accurate classification, compared to more sophisticated learning methods (see Table \ref{tab:table1}). Finally, by passing the Euclidean coordinates directly to the ML models and thereby delegating  computation of the entire pairwise LOT approximate distance to  efficient downstream methods, the LOT approximation has a large storage advantage over traditional optimal transport techniques in ML.

\section{Object Classification with LOT \label{sec:objects}}

To demonstrate the efficacy of the LOT framework, we now focus exclusively on the task of jet tagging, that is, distinguishing one type of jet from another. In addition to being an important tool in experimental analyses, jet tagging serves as an ideal playground to test new machine learning ideas in the realm of both supervised classification and unsupervised clustering. Given that optimal transport quantifies the similarity between the energy flows of two jets, the hope is that the metrics can effectively capture the differences among a variety of jet types. For the purposes of this application, we take an event to consist of a single jet and consider the flow of $p_T$ associated with particles in the jet.

Here, we consider five types of jets: single-pronged QCD (quark or gluon) jets, two-pronged boosted W boson jets, three-pronged boosted top quark jets, two-pronged boosted Higgs boson jets, and two-pronged boosted jets from a hypothetical new particle. This new Beyond-Standard-Model (BSM) particle $\phi$ is taken to be a scalar transforming in the ${\bf 6}$ representation of $SU(3)_C$ and carrying electromagnetic charge $+ \frac{1}{3}$; we consider a benchmark mass of  $m_\phi = 100$ GeV with a width of $\Gamma_\phi = 2$ GeV. It couples equally to all quark pairs that respect charge conservation. We calculate the Feynman rules for this BSM particle $\phi$ using \textsc{FeynRules} \cite{Alwall_2014}.

Instead of examining all possible pairwise combinations, we narrow our analysis to the following seven pairs: W \textit{vs} QCD, t \textit{vs} QCD, t \textit{vs} W, H \textit{vs} QCD, H \textit{vs} W, BSM \textit{vs} QCD, and BSM \textit{vs} W. For the most part, these comparisons could be thought of as treating both QCD and W boson jets as backgrounds, whereas top, Higgs boson, and BSM jets are treated as signals. The W \textit{vs} QCD pair is introduced as a benchmark for the performance of the other six tagging tasks, as well as for a meaningful comparison with the results obtained in \cite{Komiske:2019fks}.  

We generate proton-proton collision events using \textsc{madgraph} 2.6.7 \cite{Alwall_2014} at $\sqrt{s} = 14$ TeV, where the two-pronged boosted Higgs boson jets are generated via $q \bar q \to Z(\to \nu \bar \nu) + H(\to b \bar b)$, and the BSM jets through $q \bar q \to \phi \bar \phi$; all other SM jets are created via pair production. The BSM (anti)particle subsequently decays to two quarks. The matrix elements are then fed into \textsc{Pythia} 8.243 \cite{Sj_strand_2015}, with hadronization and multiple particle interactions switched on using default tuning and showering parameters. No detector simulation is included. Afterwards, we cluster the jets in \textsc{FastJet} 3.3.2 \cite{Cacciari_2012} using the anti-$k_T$ algorithm with a jet radius of 1.0, where at most two jets with $p_T \in [500, 550]$ GeV and $|y| < 1.7$ are kept.

To remove any artificial difference in the energy flows of the produced jets, every jet is preprocessed by boosting and rotating to center the jet four-momentum and vertically align the principal component of the constituent $p_T$ flow in the rapidity-azimuth plane using the \texttt{EnergyFlow} package \cite{Komiske_2018,Komiske_2019,Komiske:2019asc,Komiske:2019fks,Komiske:2019jim}. 

In order to have a unified framework for the seven comparison tasks, we work with a single choice of reference jet. The reference jet has a total $p_T$ of 525 GeV and 225 constituent particles, each with the same amount of $p_T$ evenly distributed on a $15 \times 15$ grid with   $|y| \leq 1.7$ and $|\phi| \leq \frac{\pi}{2}$. This corresponds to an isotropic distribution on the cylinder; note that related reference distributions were explored in \cite{Cesarotti:2020hwb} for the purposes of defining the event isotropy variable. We have also tried other reference jets and the resulting LOT approximation does not show any material difference compared to what is obtained from the uniform reference jet. Furthermore, as we justify rigorously in the appendix,  the LOT approximation with a uniform reference jet can be seen as an approximation of $W_{2,\R}$, the transport metric with base $\R$, which approximates the original 2-Wasserstein metric at large and small distances; see equation (\ref{lottoot}). For this reason, we will often refer to the LOT approximation as the LOT pseudo-distance in what follows.

We first normalize the $p_T$ of all jets to unity before using the Python Optimal Transport library \cite{flamary2017pot} to compute the exact OT distance between a given jet and the reference jet, with the cost being the Euclidean distance squared in the rapidity-azimuth coordinate.\footnote{This normalization step obviates the need to modify the OT distance with an additional difference term as in \cite{Komiske:2019fks}. For jet samples in the $p_T$ range explored here, we found that simple machine learning algorithms exhibit comparable or slightly better performance when using exact OT-W1 or OT-W2 distances computed between normalized jets, compared to EMD distances computed between non-normalized jets.}

 Once we have this OT distance in hand, we proceed to calculate the linear embedding for each jet using the method in Section 2. Later we recover the approximate LOT pseudo-distance between any two jets from the weighted $\ell^2$ distance between their Euclidean coordinates, which we refer to as their LOT coordinates. (Note that, due to the fact that we choose our reference jet so that all particles have equal energy, the weighted $\ell^2$ norm reduces to a classical $\ell^2$ norm in our setting.)   

Figure \ref{artificial_jets} shows the optimal energy movements between two sample QCD jets and between sample QCD and W jets using the OT-W2 distance and the LOT-W2 approximation, respectively. All jets are normalized to have unit $p_T$ before computing both metrics. In visualizing the OT-W2 metric, points in the $y$-$\phi$ plane represent constituent particles, with sizes proportional to their $p_T$; the darkness of the lines connecting points in the two jets indicate how much $p_T$ is moved from one particle to another. In visualizing the LOT pseudo-distance, vectors located at each particle in the reference jet indicate the {\it difference} between movement of $p_T$ from that particle in the reference jet to particles in the respective sample jets. In each case, the total distance between the two jets is also shown. These examples illustrate the qualitative properties of both metrics applied to simulated events: in the case of OT-W2, large OT distances correspond to the movement of significant amounts of energy between particles widely separated in the ground metric, while large LOT pseudo-distances correspond to very different transport plans between the reference jet and the respective particles. We observe that the LOT-W2 pseudo-distance is numerically close to the exact OT-W2 distance, consistent with the bounds from inequality (\ref{lottoot}).

\begin{figure}
\includegraphics[scale=0.11]{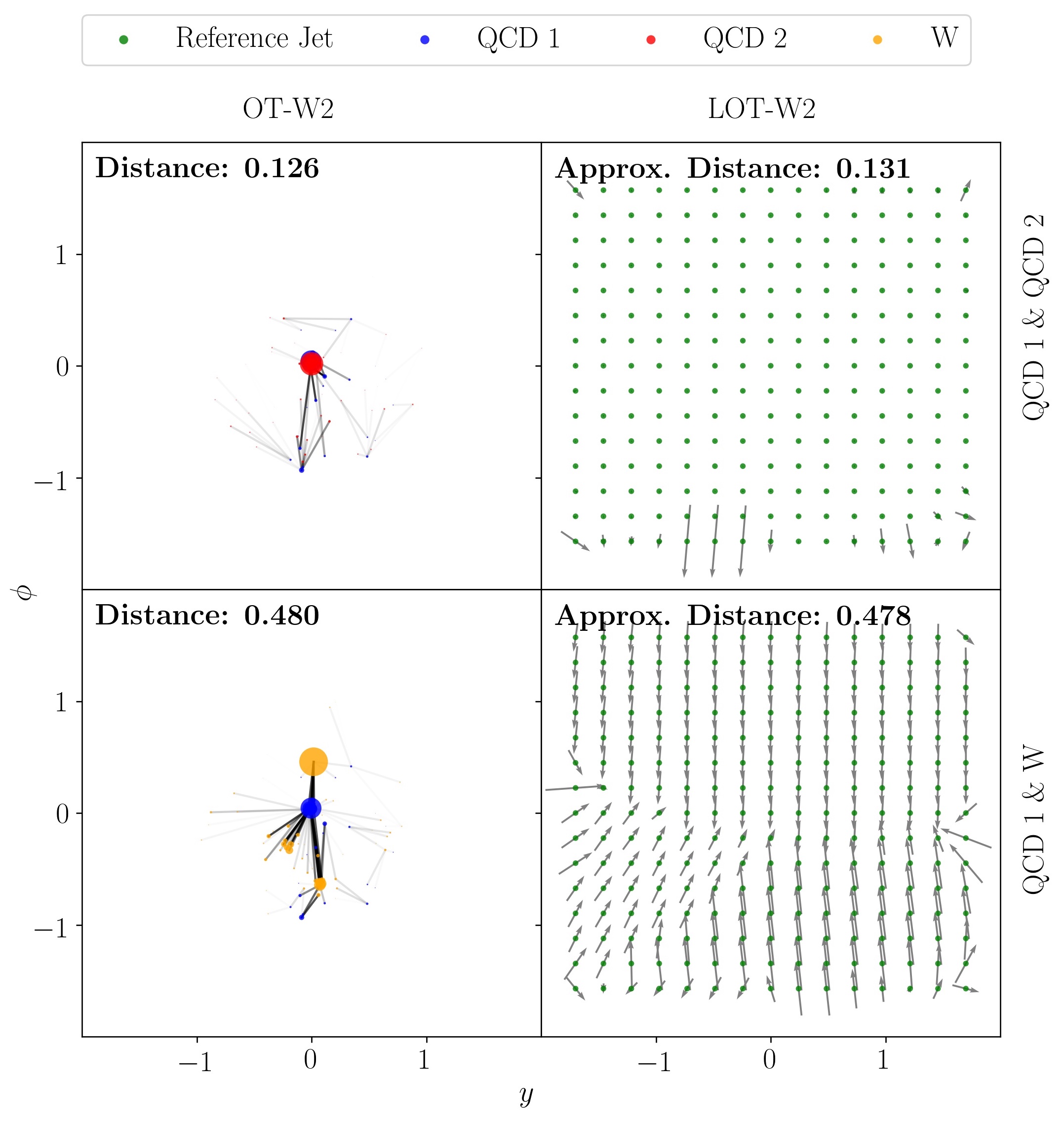}
\caption{\label{artificial_jets} 
\textit{Upper left:} The optimal movement to rearrange one QCD jet (red) into another (blue) using the exact OT-W2 metric.
\textit{Upper right:} The optimal movement to rearrange the same two QCD jets using LOT-W2.
\textit{Lower left:} The optimal movement to rearrange a W jet (orange) into a QCD jet (blue) using the exact OT-W2 metric.
\textit{Lower right:} The optimal movement to rearrange the same QCD and W jets using LOT -W2.}
\end{figure}

\section{Machine Learning with LOT \label{sec:ml}}

Once we assign  a LOT   coordinate to each jet, the inputs for jet tagging become standardized, enabling the application of a large pool of simple machine learning algorithms. Linear discriminate analysis (LDA), $k$-nearest neighbor (kNN), and support vector machine (SVM) are among many suitable algorithms for classification. Such a meaningful jet representation also makes it possible to try unsupervised clustering algorithms where we leave the model itself to assign a label for each jet. One simple example is $k$-medoids clustering. Though relatively limited in performance, all the above-mentioned traditional models have important advantages over  neural networks. They are more computationally economic, have fewer hyper-parameters to tune, and offer better human interpretability. Most of them are also off-the-shelf functions implemented in the python package \texttt{scikit-learn} \cite{scikit-learn}, making their adoption easier in practice. In our analysis, we use all four aforementioned machine learning models to either classify or cluster the jets. 

The simple supervised classifier $k$-nearest neighbor (kNN) \cite{1053964} relies on a majority vote of one's closest $k$ neighbors in the training set to determine the class membership of the new data point. Here $k$ is a model hyper-parameter to be tuned. We test $k$ in the range from 10 to 1000 with an increment of 10. Since kNN relies only on a notion of pairwise distance, it serves as a good probe to check whether our LOT approximation sufficiently captures the difference among various jet types while at the same time adequately reflecting the similarity within one specific type. The simplicity in understanding kNN and its reliance only on pairwise distances between events contribute to its adoption in the original EMD paper \cite{Komiske:2019fks}.  

A more sophisticated model, the support vector machine (SVM) \cite{Vapnik:1995}, lifts the inputs into a high-dimensional space and finds an optimal hyperplane to best separate the data. Key to SVM is the choice of a kernel function. Here we use the common rbf kernel $\exp[- \gamma d(x, x')^2]$, where $d(x, x')$ is the LOT pseudo-distance between the two data points and $\gamma$ is a tunable hyper-parameter controlling how much influence a single training example has. A high $\gamma$ suggests that only nearby points are considered. Another hyper-parameter of the model $C$ regulates the strength of the penalty term when a sample is misclassified, where a high value implies that nearly all training examples need to be classified correctly. In our analysis, we let both $C$ and $\gamma$ run from $10^{-5}$ to $10^5$ again with an increment of 10. Thus, there are $11 \times 11 = 121$ pairs of hyper-parameters and the model needs to be run for 121 times to determine the best choice.

Since both SVM and kNN involve hyper-parameter tuning, they are relatively time-consuming to train for large datasets. In contrast, linear discriminate analysis (LDA) \cite{Fisher:1936} has closed-form solutions with no hyper-parameter, making it an attractive model for a quick first look into the data. With the assumptions that the input data is Gaussian and the Gaussian for each class shares the same covariance matrix, LDA projects the input high-dimensional data onto a direction that is most discriminative, denoted as the LDA direction. Here we use LDA both as a classifier and as a tool for visualization, a point to be elaborated later.  

For unsupervised learning, we choose as a first try $k$-medoids clustering \cite{Kaufman:1987} implemented in the python package \texttt{pyclustering} \cite{Novikov2019}. The goal of the model is to partition the dataset so that the distance between points labeled to be in a cluster and the point designated as the center of that cluster is minimized. Note that the centers, called medoids, are chosen from actual data points. For the present application, the model is asked to group the unlabeled data into $k=2$ clusters. Then, the true labels are uncovered. The cluster with a higher percentage of signal jets is denoted as the signal cluster, whereas the other is designated as the background cluster. We also retrieve the true labels of the two picked medoids. Ideally, the true label of the medoid should be the same as the label of its own cluster. If not, we prefer the cluster's label. We then assign all jets in the signal cluster as signals, and those in the background cluster as background jets. This assignment is compared with the ground truth to assess the performance of our clustering model. Strictly speaking, the model is semi-supervised, for we need the true labels to decide which cluster is the signal cluster. A more detailed discussion of $k$-medoids and its performance will be given in a later paragraph. 

For every comparison task, we create two balanced datasets, each with about $50 \%$ signal jets. The smaller one, named the sample dataset, consists a total of 10,000 jets and is mainly used for picking the best hyper-parameters, though it also constitutes a complete analysis in its own right. The full dataset, on the other hand, has 140,000 jets in total, and is used to assess the model performance and draw the final conclusions. 

For the two classifiers kNN and SVM, the sample dataset is further divided into a training sample of 5000 jets, a validation sample of 2500 jets used to decide the best hyper-parameters, and a test sample of 2500 jets. The full dataset is split into a training set of 100k jets and a test set of 40k jets for these two models. For LDA, thanks to its high efficiency, we train and test on both the sample dataset (training sample size = 8000, test sample size = 2000; validation sample is not needed since there's no hyper-parameter for LDA) and the full dataset (training set size = 100k, test set size = 40k), which amounts to two separate, identical analyses. The $k$-medoids algorithm has only been applied to the sample dataset due to its computational intensity, and in this case, all 10k jets are fed into the model at once for clustering.  
 
\begin{figure*}
    \centering
    \includegraphics[scale=0.165]{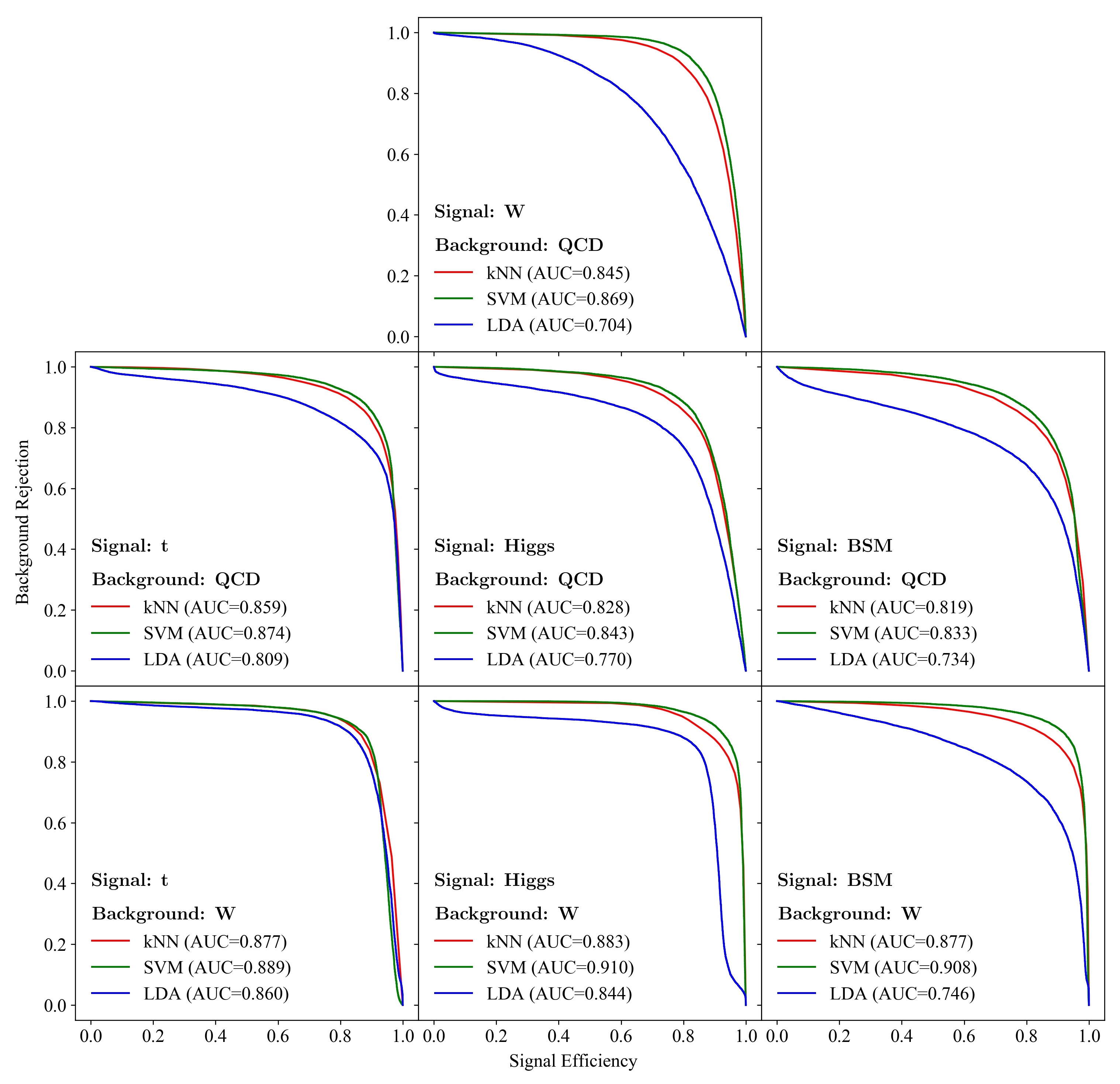}
    \caption{ROC curves for the seven jet tagging tasks evaluated on the full test datasets of 40k jets. The $x$ coordinate shows the signal efficiency rate and the $y$ coordinate gives the background rejection rate.}
    \label{fig:ROC_curves}
\end{figure*}

Fig. \ref{fig:ROC_curves} displays the receiver operating characteristic (ROC) curves of the three classifiers kNN, SVM and LDA for each of the seven comparison tasks. Also included is the Area Under the ROC Curve (AUC) which encapsulates the model performance in a single number between 0 and 1. An AUC close to 1 is most desirable, whereas a value around 0.5 suggests a random classifier, the worst-case scenario. All results are obtained on the full test datasets consisting of 40k jets, using the models trained on 100k jets with hyper-parameters, if present, picked by the sample datasets. 

To get a better sense of the model performance, we compare the AUCs of our LOT-coupled ML models for the W \textit{vs} QCD classification task with other common classifiers built in \cite{Komiske:2019fks} where the training set, though different, also contains 100k balanced W and QCD jets, and the test set contains 20k such jets.  
The model most akin to our k$_{=20}$NN-LOT is k$_{=32}$NN-EMD built upon the EMD proposed in \cite{Komiske:2019fks}, an interpolation between the OT-W1 distance and total variation norm.\footnote{Although our samples are not identical to those in \cite{Komiske:2019fks}, we apply the same prescription for simulating and preparing the samples, and our W/QCD jet samples yield results for k$_{=32}$NN-EMD compatible with \cite{Komiske:2019fks}.} The N-subjettiness ratio $\tau_2^{\beta=1}$ / $\tau_1^{\beta=1}$, introduced in \cite{Thaler_2011,Thaler_2012}, is a widely-used observable specifically designed to spot two-prong jet substructure. For the other three classifiers, namely the  Energy Flow Network (EFN) and Particle Flow Network (PFN) neural networks \cite{Komiske_2019}, and a linear classier trained on Energy Flow Polynomials (EFPs) \cite{Komiske_2018}, please refer to the original papers for more details.

\begin{center}
\begin{tabular}{ |c|c|c| } 
\hline 
 \textbf{Datasets} & \textbf{Model} & \textbf{AUC} \\
 \hline 
 \multirow{3}{*}{Our Datasets} & k$_{=20}$NN-LOT & 0.845 \\ 
 \cline{2-3}
  & SVM-LOT & 0.869 \\ 
 \cline{2-3}
  & LDA-LOT & 0.704 \\ 
 \hline
 \multirow{5}{*}{Datasets in \cite{Komiske:2019fks}} & k$_{=32}$NN-EMD & 0.887 \\ 
 \cline{2-3}
  & $\tau_2^{\beta=1}$ / $\tau_1^{\beta=1}$  & 0.776 \\
 \cline{2-3}
  & PFN & 0.919 \\
 \cline{2-3}
  & EFPs & 0.917 \\
 \cline{2-3}
  & EFN & 0.904 \\ 
  \hline 
\end{tabular}
\end{center} 

Not surprisingly, the neural networks obtain the best performance. But the four optimal transport inspired models (three with LOT and one with EMD) are on a par with these state-of-the-art complex classifiers, and they significantly outperform the N-subjettiness observable (with the single exception of the exceptionally simplistic LDA). More pertinent to our current investigation is the observation that models coupled with LOT-W2 approximation perform as well as those using the exact EMD metric. The AUCs of kNN-LOT and SVM-LOT are close to the AUC of kNN-EMD, suggesting that it does not make much difference for jet tagging whether we use the exact OT metric or its linearized version. Yet on the practical level, the LOT approximation has a significant advantage over the exact OT metric. The computation of the LOT coordinates for 140k jets only takes about 10 minutes on a desktop computer, whereas it is infeasible to compute the full exact OT matrix of pairwise distances on the same computer and still requires significant time on a cluster.

\begin{table*}[htbp]
\caption{Results for the seven jet tagging tasks using four different machine learning models coupled with the LOT coordination.}
\begin{center}
\begin{adjustbox}{scale=0.95, center}
\begin{tabular}{|c|c|c|c|c|c|c|c|c|c|}
    \hline
    \multirow{2}{*}{Model} & \multirow{2}{*}{Dataset} & & \multicolumn{7}{|c|}{Comparison Task}\\
    \cline{3-10}
    & & & W $vs$ QCD & t $vs$ QCD & t $vs$ W & H $vs$ QCD & H $vs$ W & BSM $vs$ QCD & BSM $vs$ W\\
    \hline
    \hline
    \multirow{8}{*}{LDA} & \multirow{4}{*}{Sample Dataset} & \textbf{AUC} & \textbf{0.6896} & \textbf{0.7863} & \textbf{0.8464} & \textbf{0.7642} & \textbf{0.7865} & \textbf{0.7158} & \textbf{0.7244}\\
    & & TPR & 0.6926 & 0.7746 & 0.7886 & 0.7378 & 0.7762 & 0.6713 & 0.6562\\
    & & FPR & 0.3133 & 0.2020 & 0.0958 & 0.2095 & 0.2032 & 0.2397 & 0.2074\\
    \cline{3-10}
    & & Approx. Run Time & \multicolumn{7}{|c|}{several seconds}\\
    \cline{2-10}
    & \multirow{4}{*}{Full Dataset} & \textbf{AUC} & \textbf{0.7041} & \textbf{0.8077} & \textbf{0.8573} & \textbf{0.7703} & \textbf{0.8443} & \textbf{0.7337} & \textbf{0.7455}\\
    & & TPR & 0.7156 & 0.7969 & 0.7957 & 0.7661 & 0.8254 & 0.7549 & 0.6804\\
    & & FPR & 0.3075 & 0.1815 & 0.0812 & 0.2255 & 0.1368 & 0.2874 & 0.1894\\
    \cline{3-10}
    & & Approx. Run Time & \multicolumn{7}{|c|}{several seconds}\\
    \hline
    \hline
    \multirow{10}{*}{SVM} & \multirow{4}{*}{Sample Dataset} & \textbf{AUC} & \textbf{0.8410} & \textbf{0.8630} & \textbf{0.8751} & \textbf{0.8349} & \textbf{0.8831} & \textbf{0.8239} & \textbf{0.8806}\\
    & & TPR & 0.8148 & 0.8929 & 0.8333 & 0.8006 & 0.8750 & 0.8582 & 0.9090\\
    & & FPR & 0.1327 & 0.1669 & 0.0831 & 0.1308 & 0.1088 & 0.2104 & 0.1478\\
    \cline{3-10}
    & & Approx. Run Time & \multicolumn{7}{|c|}{2 hours}\\
    \cline{2-10}
    & \multirow{4}{*}{Full Dataset} & \textbf{AUC} & \textbf{0.8687} & \textbf{0.8780} & \textbf{0.8805} & \textbf{0.8426}& \textbf{0.9100} & \textbf{0.8331} & \textbf{0.9077}\\
    & & TPR & 0.8451 & 0.8873 & 0.8365 & 0.8185 & 0.9103 & 0.8471 & 0.9191\\
    & & FPR & 0.1077 & 0.1313 & 0.0755 & 0.1332 & 0.0904 & 0.1808 & 0.1037\\
    \cline{3-10}
    & & Approx. Run Time & \multicolumn{7}{|c|}{6 hours}\\
    \cline{2-10}
    & \multirow{2}{*}{Hyperparameters} & $C$ & 1.0 & 1.0 & 10.0 & 1.0 & 1.0 & 1.0 & 1.0\\
    \cline{3-10}
    & & $\gamma$ & 100.0 & 100.0 & 10.0 & 100.0 & 100.0 & 100.0 & 100.0\\
    \hline
    \hline
    \multirow{9}{*}{kNN} & \multirow{4}{*}{Sample Dataset} & \textbf{AUC} & \textbf{0.8191} & \textbf{0.8450} & \textbf{0.8659} & \textbf{0.8203} & \textbf{0.8628} & \textbf{0.8026} & \textbf{0.8361}\\
    & & TPR & 0.7741 & 0.8164 & 0.8040 & 0.7975 & 0.8295 & 0.8172 & 0.8241\\
    & & FPR & 0.1358 & 0.1264 & 0.0723 & 0.1568 & 0.1038 & 0.2120 & 0.1520\\
    \cline{3-10}
    & & Approx. Run Time & \multicolumn{7}{|c|}{15 minutes}\\
    \cline{2-10}
    & \multirow{4}{*}{Full Dataset} & \textbf{AUC} & \textbf{0.8455} & \textbf{0.8601} & \textbf{0.8735} & \textbf{0.8280} & \textbf{0.8831} & \textbf{0.8192} & \textbf{0.8772}\\
    & & TPR & 0.8033 & 0.8217 & 0.8156 & 0.8040 & 0.8566 & 0.8261 & 0.8836\\
    & & FPR & 0.1123 & 0.1014 & 0.0686 & 0.1479 & 0.0905 & 0.1876 & 0.1292\\
    \cline{3-10}
    & & Approx. Run Time & \multicolumn{7}{|c|}{4 hours}\\
    \cline{2-10}
    & Hyperparameter & $k$ & 20 & 40 & 10 & 20 & 20 & 10 & 20\\
    \hline
    \hline
    \multirow{9}{*}{$k$-medoids} & \multirow{9}{*}{Sample Dataset} & \textbf{AUC} & \textbf{0.6797} & \textbf{0.8096} & \textbf{0.8074} & \textbf{0.7689} & \textbf{0.8028} & \textbf{0.7622} & \textbf{0.6698}\\
    & & TPR & 0.7947 & 0.9282 & 0.6583 & 0.8374 & 0.6835 & 0.8837 & 0.5216\\
    & & FPR & 0.4354 & 0.3089 & 0.0436 & 0.2996 & 0.0778 & 0.3592 & 0.1821\\
    \cline{3-10}
    & & Signal Percentage & (63.78\%, & (74.70\%, & (94.00\%, & (73.60\%, & (90.11\%, & (71.05\%, & (74.81\%,\\
    & & (sig, bkg) & 25.97\%) & 9.27\%) & 27.02\%) & 18.81\%) & 26.24\%) & 15.33\%) & 37.75\%)\\
    \cline{3-10}
    Clustering & & Clusters' Size & (6118, & (6159, & (3565, & (5682, & (3861, & (6211, & (3549,\\
    & & (sig, bkg) & 3882) & 3841) & 6435) & 4318) & 6139) & 3789) & 6451)\\
    \cline{3-10}
    & & Medoids True Labels & \multirow{2}{*}{(1, 0)} & \multirow{2}{*}{(0, 0)} & \multirow{2}{*}{(1, 0)} & \multirow{2}{*}{(1, 0)} & \multirow{2}{*}{(1, 1)} & \multirow{2}{*}{(1, 0)} & \multirow{2}{*}{(1, 0)}\\
    & & (sig: 1, bkg: 0) & & & & & & &\\
    \cline{3-10}
    & & Approx. Run Time & \multicolumn{7}{|c|}{30 minutes}\\
    \hline
\end{tabular}
\end{adjustbox}
\end{center}
\label{tab:table1}
\end{table*}

Table \ref{tab:table1} summarizes the results obtained for all seven comparison tasks, with complete, independent analyses done both on the sample datasets and the full datasets. In addition to AUC, we also report the True Positive Rate (TPR) and False Positive Rate (FPR), where the TPR is the same as the signal efficiency, and the FPR equals to one minus the background rejection. A TPR near 1 and a FPR close to 0 are preferable. For SVM and kNN, we also include the hyper-parameters chosen by the sample datasets. The results for $k$-medoids are harder to interpret, so we defer a full discussion to a later paragraph. 

Also included in the table is the approximate run time for each task, performed on an iMac with 3.6 GHz 8-Core Intel Core i9 and 16 GB memory. The longest analysis takes no more than 10 hours, which, when combined with the extra few minutes for calculating the LOT coordinates, is quite manageable. LDA in particular only takes seconds to process the full datasets and in this light its classification results are surprisingly good. In addition, models performed on the sample datasets require as few as 2 hours for a full scan of hundreds of possible combinations of hyper-parameters. Competitive classification performance coupled with efficient computational time suggests that the linearized optimal transport metric may play a role in event classification alongside the exact OT metric, complex neural networks, and traditional handpicked observables.
 
Given that the sample datasets constitute complete analyses on their own rights, we can compare their results with those obtained using the full datasets. In general, model performance naturally gets better with more training data, but we observe that the increase in performance going from 10k jets to 140k jets is perhaps not significant enough to justify the extra computational resources needed. Since the numbers quoted for AUC, TPR and FPR are only intended as general performance evaluations rather than precise measures, the fluctuations in these numbers can be safely ignored and we therefore conclude that a dataset of 10,000 jets (with as few as five thousands for training) is already enough to assess the overall quality of the model and the underlying metric.  

Some general features can be immediately read off from the table. Whichever jets we compare, SVM always gives the best classification performance with AUCs around 0.9, approaching the performance of neural networks. This suggests that jets represented in their LOT coordinates are indeed very well separated by a hyperplane in some high-dimensional feature space, which in turn demonstrates the fitness of the approximate metric itself. Except for t \textit{vs} W jets classification, the hyper-parameters chosen for SVM via the validation process are all the same, with $C = 1$ and $\gamma = 100$ where 1 happens to be the default value for C in \texttt{scikit-learn}. It means that the model uses only a reasonable amount of regularization and thus a relatively smooth decision surface is drawn. On the other hand, a $\gamma$ of 100 is considered large, indicating that only nearby samples can have an influence on the classification of a new point. 

This latter observation is consistent with what is suggested by the hyper-parameter $k$ picked by kNN. All seven comparison tasks prefer small $k$ values less than 50, which means that to determine the type of an unknown jet we need to look no further than its closest 50 neighbors. If LOT does not place same-type jets near each other as desired, then models with hyper-parameters preferring locality won't be able to achieve such satisfying classification performances. Therefore, the hyper-parameters picked by SVM and kNN provide an indirect evidence for the suitability of the optimal transport metric------it indeed groups jets of the same type near each other and separates those of different types. We will later turn this speculation into more convincing and intuitive visualization. 

Among the seven jet tagging tasks, kNN and SVM both have the best performance in distinguishing Higgs boson jets from W boson jets and are least capable of separating BSM jets from QCD jets. This is mainly caused by a relatively high false positive rate, meaning that the models have a tendency to wrongly classify QCD jets as BSM jets. The same reason applies to LDA when it performs poorly on W \textit{vs} QCD classification relative to other tasks. For each type of signal jets (t, H, or BSM), all three classification models perform better when the background is W jet rather than QCD jet.

We now focus on the $k$-medoids clustering algorithm, which is only analyzed on the sample datasets due to computational limitations. Given that unsupervised learning is inherently more difficult than supervised learning, it's not surprising to see the performance of $k$-medoids algorithm to be inferior to that of kNN or SVM. But even then, except for the W \textit{vs} QCD and BSM \textit{vs} W tasks, the AUCs of $k$-medoids are all above 0.75, on a par with the supervised learning models analyzed on the sample datasets. The clustering algorithm even shows superior performance compared to LDA for most tagging tasks. This remarkable achievement again points to the merit of the underlying approximate LOT distance and is encouraging for the further exploration of optimal transport applications to unsupervised learning algorithms.

It should be noted that AUC is not the only gauge of model performance. Especially in the case of $k$-medoids clustering, we also need to take a look at other indicators to map a more complete picture. Beside examining the TPR and FPR, we also like to know more about the properties of the two clusters outputted by the algorithm. If the model is perfect, then each cluster should contain only signal jets or only background jets. The purity of the two clusters is given in the second row of $k$-medoids clustering in the table, where we record the signal percentage (defined as the number of signals in the cluster divided by the total number of jets in that cluster) in the signal cluster and the background cluster, respectively. By definition, the signal cluster is the group with a majority of signal jets, which, if pure, should have a signal percentage of $100\%$. Similarly, a pure background cluster should have $0\%$ signal percentage. Notice that the sum of the signal percentage of the two clusters does not necessarily equal to 1 (but in the ideal case it is). The worst-case scenario is to have the signal percentage of both clusters close to $50\%$. A quick look at the second row at least qualitatively confirms that the AUC of the task is indeed higher whenever we have two purer clusters, with the best AUC obtained for t \textit{vs} QCD clustering which has a signal percentage of $74.70\%$ for the signal cluster and only $9.27\%$ for the background cluster. 

The size of the clusters also reveals how well the model performs. Ideally, the result would be two clusters with equal size, that is, each with 5000 jets, since the data itself is balanced. Here the best result we have is for H \textit{vs} QCD task, where the Higgs cluster has 5682 jets and the QCD cluster has a total of 4318 jets. But in general, the two clusters are not well balanced. In the worst case, the W cluster has $81.77\%$ more jets than the BSM cluster, and it does correspond to the lowest AUC score. 

In theory, the two medoids should be the most representative jet for the clusters they respectively belong. Since the medoids are actual data points, we can uncover their true labels and check whether they agree with the type of the cluster they're assigned to. Only the two tasks, t \textit{vs} QCD and H \textit{vs} W, give conflicting answers. For the t \textit{vs} QCD clustering, the two chosen medoids are both background QCD jets. Thus the signal top cluster acquires a QCD jet as its representative. The situation is reversed for the H \textit{vs} W task where now the background W cluster elects a signal Higgs jet as its exemplar. Nevertheless, both tasks enjoy high AUC scores, which suggests that the true labels of the medoids might not have a direct influence on model performance.  

The general message here is that AUC, though powerful and straighforward, is not enough to assess the performance of an algorithm; other indicators are required to gain a fuller appreciation of the strength and weakness of the model, both for clustering and for classification.    

Lastly, we use LDA to visualize jets and aid understanding of the LOT approximation and its associated Euclidean embedding. Our approach follows work by Wang, et. al. \cite{wang2013linear}, which introduced the LOT framework and applied it to visualization tasks, such as discriminating nuclear chromatin patterns in cancer cells. Given the $225 \times 2$ linearized coordinate for each jet, we first stack the list of the second coordinate $\phi$ at the end of the list of the first coordinate $y$ and reshape the coordinate to be $450 \times 1$, which is then fed into a LDA model for the projection of the 450 coordinates onto one single most discriminative direction (denoted as the LDA direction). This allows us to represent every jet as one single point on the LDA direction for easy visualization. Fig. \ref{fig:LDA_sigmas_combined_Wt10000} shows such projection for the 10000 jets in the t \textit{vs} W sample dataset, which enjoys the highest AUC among the seven tasks with the LDA classifier. A clear separation between W and top jets can be seen, with the majority of W boson jets grouped towards the left end of the LDA direction and most top jets towards the right end, explaining the good performance of the LDA classifier for this task. 

It is enlightening to see how jets vary along the chosen LDA direction. To this end, we first select the jet whose 1-dimensional projected LDA coordinate has a value closest to the mean of all LDA coordinates in the dataset and denote it as the mean jet. We then compute the standard deviation of the dataset. Now jets whose LDA coordinates are up to 3 sigmas away from the mean jet are displayed in Fig. \ref{fig:LDA_sigmas_combined_Wt10000}. We observe a clear tendency of particles spreading more on the $y$-$\phi$ plane as we move from the left end of the LDA direction to the right end, i.e., from negative sigmas to positive sigmas, corresponding well to our intuition that top jets are more smeared and tend to have a three-pronged structure.   

\begin{figure*}
    \centering
    \includegraphics[scale=0.075]{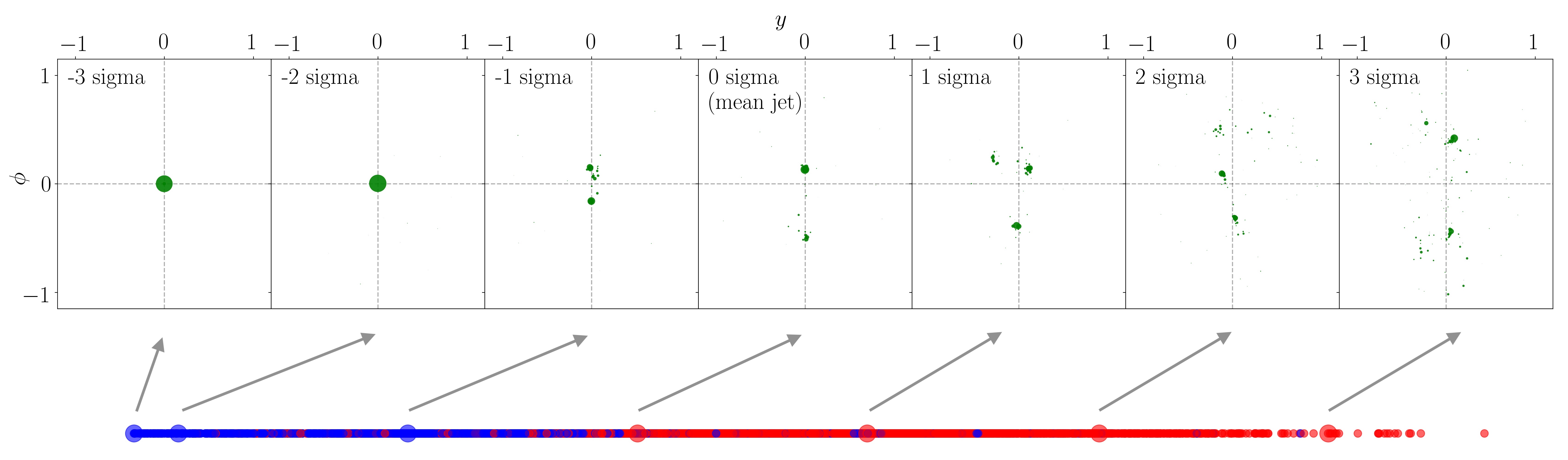}
    \caption{\textit{Bottom:} Projection of the LOT coordinates of 10,000 jets in the sample dataset onto the LDA direction chosen by the model. Blue dots represent W boson jets and red dots refer to top jets. The seven larger dots represent jets whose LDA coordinates are $-3,-2,-1,0,1,2,3$ sigma away from the mean jet (starting from the left). 
    \textit{Top:} The energy flow in the rapidity-azimuthal plane of the seven jets chosen in the bottom plot respectively. The intersection of the dashed lines shows the location of the origin in the $y$-$\phi$ plane.}
    \label{fig:LDA_sigmas_combined_Wt10000}
\end{figure*}

As another illustration, we examine more closely how the OT-W2 metric rearranges the $p_T$ of one jet to make it look like another, as shown in Fig. \ref{fig:LDA_1to1_Wt10000}. Here we first select the rightmost top jet $t^1$ and the leftmost W boson jet $W^1$ in the bottom plot of Fig. \ref{fig:LDA_sigmas_combined_Wt10000}. We then compute the exact 2-Wasserstein optimal transportation matrix $\gamma_{ij}$, which instructs how much of $p_T$ is moved from particle $i$ in jet $W^1$ (denoted as $W^1_i$) to particle $j$ in jet $t^1$ (denoted as $t^1_j$). To interpolate between the two extreme jets, we create a new jet that depends on an interpolation parameter $\alpha \in [0, 1]$, where $\alpha = 0$ outputs a jet identical to $W^1$ and $\alpha = 1$ recovers the $t^1$ jet. This new artificial jet$^\alpha$ contains $i \times j$ particles, each with 
\begin{align} \nonumber
   p_T^{\alpha} &= \gamma_{ij},  \\
   y^{\alpha} &= (1 - \alpha) \times y(W^1_i) + \alpha \times y(t^1_j), \\ \nonumber
   \phi^{\alpha} &= (1 - \alpha) \times \phi(W^1_i) + \alpha \times \phi(t^1_j),
\end{align} 
where $y(W^1_i)$ is the $y$ coordinate of the $i$th particle in jet $W^1$, and likewise for the others. From the perspective of optimal transport theory, this artificial jet is precisely the 2-Wasserstein geodesic between the jets. Several values of $\alpha$ are picked in Fig. \ref{fig:LDA_1to1_Wt10000} so as to show a few representatives of the interpolated jets and help us to understand intuitively the $p_T$ movement by the OT-W2 metric. This interpolation technique may prove relevant to the fast simulation of collider events, insofar as it allows interpolation between real events.

\begin{figure*}
    \centering
    \includegraphics[scale=0.085]{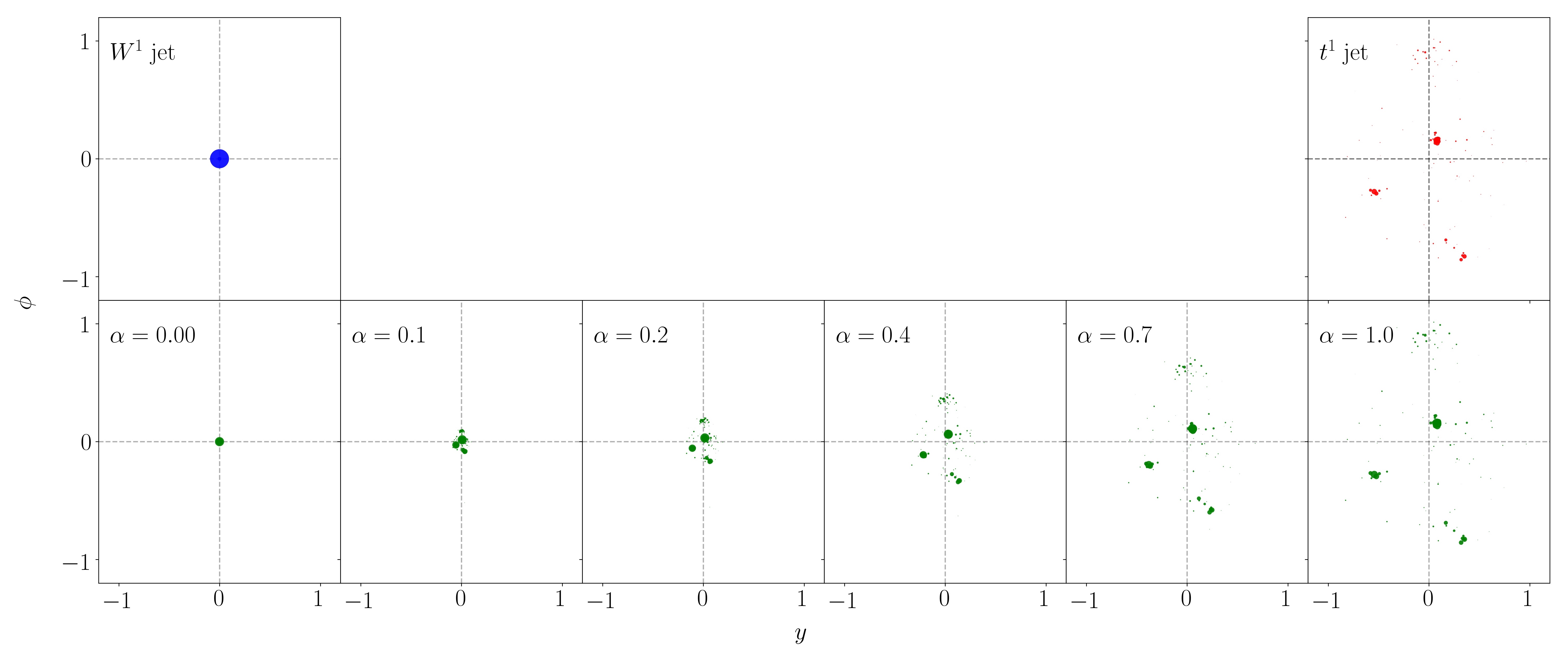}
    \caption{The OT-W2 movement of $p_T$ to rearrange the leftmost W boson jet $W^1$ (blue) into the rightmost top jet $t^1$ (red) in the sample dataset. The intermediate green plots show artificial jets created via the interpolation parameter $\alpha$. When $\alpha = 0$ and $1$, the jets are respectively identical to $W^1$ and $t^1$ up to visualization. Again the intersection of the dashed lines shows the location of the origin. 
    \label{fig:LDA_1to1_Wt10000}}
\end{figure*}

The above visualizations provide useful insight into the performance of the LOT approximation and the machine learning model coupled to it, offering a useful intermediary between analytic kinematic variables and deep neural networks.

\section{Conclusion \label{sec:conc}}
The theory of optimal transport offers a new perspective on the traditional problems of collider physics, beginning with the introduction of the OT-based Energy Mover's Distance in \cite{Komiske:2019fks}. But the practical value of exact OT metrics as competitors to specialized variables and deep neural networks is limited by the need to determine $\mathcal{O}(N_{\rm evt}^2)$ computationally expensive OT distances between $N_{\rm evt}$ events. In this paper we have introduced an efficient approximation scheme for computing optimal transport distances in collider events using a linear optimal transport approximation to the 2-Wasserstein distance. This entails computing the exact OT distance between each event and a reference jet containing $n$ particles; the corresponding transport plan provides a map from the event to a vector in $n$-dimensional Euclidean space. The approximate LOT distance between two events is then obtained by computing a simple weighted $\ell^2$ distance between the corresponding $n$-vectors, so that only $\mathcal{O}(N_{\rm evt})$ OT distances and  $\mathcal{O}(N_{\rm evt}^2)$ $\ell^2$ distances are required. This makes the calculation of approximate OT distances between collider events in a typical sample accessible to a desktop computer. Furthermore, we have proved that this LOT approximation converges to a true metric on the space of collider events in the continuum limit.

The Euclidean embedding furnished by our approximation scheme makes it a natural input to simple machine learning algorithms that require more than the pairwise distance between events, such as LDA. We have demonstrated the value of the LOT framework for jet tagging in a number of classification tasks, illustrating both the relative computational efficiency (compared to exact OT approaches) and interpretability (compared to deep neural networks) of our approach. The two classifiers kNN and SVM coupled with the LOT approximation achieve high performance on a level comparable to both the exact OT approach and complex neural networks, while significantly outperforming the traditional N-subjettiness variable. The choice of the hyper-parameters of the two models further confirms the effectiveness of the approximate LOT distance in capturing the difference among various jet types. As a quick first look into the datasets, LDA performs surprisingly well and provides an intuitively clear visualization method.  The good performance of the $k$-medoids clustering algorithm is encouraging for further explorations of the application of the LOT framework to tasks beyond supervised learning, including clustering and anomaly/novelty detection. Finally, the similarity in the performance of the sample datasets and the full datasets suggests that only as few as 10k jets are required to have an estimate on the quality of the model and the underlying metric, further reducing the computational cost.       

There are a wide variety of future directions. The computational speedup offered by the LOT approximation should make it possible to apply optimal transport methods more broadly in analyzing both simulated and actual collider data. Likewise, this speedup motivates extending LOT methods to other optimal transport metrics (such as unbalanced OT) which may be relevant to collider physics but whose application is currently limited by computational cost. To the extent that it involves the transport plan from a reference jet to an event, the approximate LOT distance shares aspects with the OT-based event isotropy variable \cite{Cesarotti:2020hwb}, and it would be interesting to investigate their relationship further. The convergence of the LOT approximation to a true metric in the continuum limit suggests it may play a role as a discrete approximation scheme in the broader geometric approach to collider observables proposed in \cite{Komiske:2020qhg}.   

More broadly, there remains much to explore at the interface between collider physics and the theory of optimal transport.  

\begin{acknowledgments}
J. Cheng would like to thank Timothy Trott for providing example FeynRules files. K. Craig would like to thank Bernhard Schmitzer and Dejan Slep\v cev for   helpful conversations about LOT. We would like to thank Jesse Thaler for comments on the manuscript. The work of T. Cai and N. Craig was supported in part by the Department of Energy under the grant DE-SC0250757. The work of J. Cheng was supported by The Create Fund, thanks to the generosity of CCS donors. The work of K. Craig was supported by National Science Foundation grant DMS-1811012 and a Hellman Faculty Fellowship.
\end{acknowledgments}

\appendix*
\section{From LOT approximation to LOT distance \label{app:proof}}

In this appendix, we prove the convergence of the LOT approximation, defined in equation (\ref{LOTgdef}), to a true metric in the continuum limit. For the sake of brevity, we will only briefly discuss the optimal transport theory underlying this result, primarily with the goal of establishing  notation. We refer the reader to the textbooks by Ambrogio, Gigli, and Savar\'e \cite{AGS}, Peyr\'e and Cuturi \cite{peyre2019computational}, Santambrogio \cite{santambrogio2015optimal}, and Villani \cite{villani2003topics} for further background.
 
 Let $\P(\Rd)$ denote the set of probability measures on $\Rd$. Given $\mu, \nu \in \P(\Rd)$, a measurable function $\bt: \Rd \to \Rd$ \emph{transports $\mu$ onto $\nu$} if $\nu(B) = \mu(\bt^{-1}(B))$ for all measureable sets $B \subseteq \Rd$. We call $\nu$ the \emph{push-forward of $\mu$ under $\bt$} and write $\nu = \bt \# \mu$. For historical reasons, it is conventional in the field of optimal transport to think of the amount of measure  $\mu$ gives to a measurable set $B$ as the \emph{mass of $B$ with respect to $\mu$} and to interpret a measurable function $\bt$ as a \emph{transport map} that \emph{rearranges the mass in $\mu$ to look like $\nu$}. Conveniently for physicists, the ``mass'' and ``energy'' notation is equivalent in natural units, and we will use the former here. Given a probability measure on a product space, for example $\bgamma \in \P(\Rd \times \Rd)$, its \emph{marginals} are given by the pushforward of the measure through the projections on each component of the product. For example, if  $\pi^2: \Rd \times \Rd \to \Rd$ is the projection onto the second component of $\Rd \times \Rd$, then $\pi^2 \# \bgamma$ is the \emph{second marginal} of $\gamma$. Finally, we  say that $\E \in \mathcal{P}(\Rd)$ has \emph{finite second moment} if $M_2(\E):=\int_\Rd |x|^2 d \E(x) < +\infty$, in which case we write $\E \in \P_2(\Rd)$. 

For any $\E, \tcE \in P_2(\Rd)$, the 2-Wasserstein distance from $\E$ to $\tcE$ is given by
\begin{align*}
 W_2(\E,\tcE) &= \min_{ \bgamma \in \Gamma(\E, \tcE)}  \left( \int_{\Rd \times \Rd} |x-y|^{2}d \bgamma(x,y)     \right)^{1/2} , \\
 \Gamma(\E, \tcE) &= \{ \bgamma \in \P(\Rd \times \Rd) : \pi^1 \# \bgamma = \E \ , \pi^2 \# \bgamma = \tcE \}  
 \end{align*}
 Note that, in the special case $\E = \sum_i \delta_{x_i} E_i$, $\tcE = \sum_j \delta_{\tx_j} \tE_j$, the above definition of the 2-Wasserstein distance coincides with that given   in section \ref{sec:lot}.
We refer to the set of transport plans $\gamma \in \Gamma(\E,\tcE)$ that achieve the minimum as the set of \emph{optimal transport plans}, which we denote by $\Gamma_0(\mu, \nu)$. Furthermore, we say that a plan $\bgamma \in \Gamma(\E,\tcE)$ is \emph{induced by a transport map} if there exists a measurable function $\bt: \Rd \to \Rd$ so that $\bgamma = (\id \times \bt) \# \E$, where $\id(x) = x$ is the identity mapping.
 
Just as we may extend the 2-Wasserstein distance from the discrete case to the case of probability measures, we may likewise extend the definition of the LOT functional, as well as define the related concept of transport metrics.
We devote particular attention to the case that the reference measure $\R$ does not give mass to sets of $(d-1)$-dimensional Hausdorff measure; in other words, the measure does not concentrate on small sets. In this case, for any $\E \in \P_2(\Rd)$, there exists a unique optimal transport plan $\brho \in \Gamma_0(\R, \E)$, and $\brho$ is induced by a transport map \cite{Gigli}. This transport map is unique $\E$-almost everywhere, and we refer to it as the \emph{optimal transport map from $\R$ to $\E$,} denoted  $\bt_\R^\E$ \cite{McCannExistence}. The function $x \mapsto \bt_\R^\E(x)$ represents where mass starting at location $x$ in the reference measure $\R$ is sent in the target measure $\E$, in order to rearrange the mass from $\R$ into $\E$, using the least amount of effort. Note that a necessary condition for such an optimal transport map to exist is that an optimal rearrangement of $\R$ to $\E$ does not \emph{split mass}; that is, all mass starting at a specific location in $\R$ must be sent to the same location in $\E$.  

Given a reference measure $\R \in \P_2(\Rd)$, which does not give mass to sets of $(d-1)$-dimensional Hausdorff measure, and measures $\E, \tcE \in \P_2(\Rd)$, the \emph{transport metric with base $\R$} is given by 
\begin{align} \label{transportmetricdef}
W_{2, \R}(\E, \tcE) = \left( \int | \bt_{\R}^{\E} - \bt_{\R}^{\tcE} |^2 d \R \right)^{1/2} . 
\end{align}
The transport metric with base $\R$ is a well-defined metric on $\P_2(\Rd)$, which can be interpreted as computing the distance between $\E$ and $\tcE$ by projecting onto the tangent plane at $\R$  \cite[Proposition 1.15]{craig2016exponential}, \cite[equation 6]{wang2013linear}, \cite[equations (7.3.2), (9.2.5), Theorem 8.5.1]{AGS}.  

 In this section, we prove that the linearized optimal transport approximation converges as the discretization of the reference measure is refined. In order to do this, we now define the LOT functional for general measures and show its relationship with the transport metric with base $\R$. Given  measures $\R, \E, \tcE \in \P_2(\Rd)$, for any $\brho \in \Gamma_0(\R, \E)$, $\tbrho \in \Gamma_0(\R, \tcE)$, there exists $\bomega \in \P(\Rd \times \Rd \times \Rd)$ so that
\begin{align}
\pi^{1,2} \# \bomega = \brho  \quad \text{ and } \quad \pi^{1,3} \# \bomega = \tbrho ,
\end{align}
where $\pi^{i,j}$ is the projection on the $i$th and $j$th components of $\Rd \times \Rd \times \Rd$; when $\R$ doesn't give mass to small sets, then $\bomega$ is unique \cite[Lemma 5.3.2]{AGS}. By disintegration of measures, there exists a family $\{ \bomega_{x_1} \in \P(\Rd \times \Rd) \}_{x_1 \in \Rd}$ so that for any  measurable function $f: \Rd \times \Rd \times \Rd \to [0,+\infty)$,
\begin{align} \nonumber
\int_\Rd \left( \int_{\Rd\times \Rd} f(x_1, x_2,x_3)d\bomega_{x_1}(x_2,x_3) \right) d \R(x_1) \\
\quad =  \int_{\Rd\times \Rd \times \Rd} f(x_1, x_2,x_3)d\bomega(x_1,x_2,x_3).
\end{align}

In this way,  for  $\R, \E, \tcE \in \P_2(\Rd)$ and $\brho \in \Gamma_0(\R, \E)$, $\tbrho \in \Gamma_0(\R, \tcE)$, the LOT functional is defined by 
 \begin{align} \label{genLOT}
&LOT_{\brho, \tbrho}(\E,\tcE) \\
&\quad = \left( \int \left| \int (x_2 - x_3) d \bomega_{x_1}(x_2,x_3) \right|^2 d \R(x_1) \right)^{1/2}. \nonumber
\end{align}
In the special case that $\R = \sum_i \delta_{y_i}R_i$, $\E = \sum_j \delta_{x_j} E_j$, and $\tcE = \sum_k \delta_{\tx_k} \tE_k$, this reduces to the LOT functional defined in section \ref{sec:lot}. Furthermore, in the special case that $\R$ does not give mass to sets of $(d-1)$-dimensional Hausdorff measure, the optimal transport plans $\brho = (\id \times \bt_\R^\E) \# \R$ and $\tbrho = (\id \times \bt_\R^\tcE) \# \R$ are unique, as is the measure $\bomega = (\id \times \bt_\R^\E \times \bt_\R^\tcE) \# \R$ and its disintegration $\bomega_{x_1} = \delta_{(\bt_\R^\E(x_2), \bt_\R^\tcE(x_3))}$. Consequently, when $\R$ does not give mass to small sets,  the LOT functional is independent of the choice of transport plans $\brho, \tbrho$, and $LOT_{\brho, \tbrho}(\E, \tcE) = W_{2, \R}(\E, \tcE)$; that is, the LOT approximation becomes a well-defined metric on the space of probability measures with finite second moment. Similarly, when $\R$ does not give mass to small sets, the LOT Euclidean  embedding can be thought of, from a geometric perspective, as the inverse of the exponential map
\begin{align}
\E \mapsto \int x_2  d \bomega_{x_1}(x_2,x_3) = \bt_\R^\E ,
\end{align} 
which is an isometric embedding from $W_{2, \R}$ to $L^2(\R)$.

We now prove that, for any sequence   $\R^N \xrightarrow{W_2} \R$, where $\R$ does not give mass to small sets,   the LOT approximation corresponding to $\R^N$ converges to the transport metric with base $\R$. Furthermore, we allow the events $\E^N$ and $\tcE^N$ to likewise vary along  convergent sequences.
 
\begin{prop} \label{convergenceLOTprop}
Consider three sequences of probability measures $\R^N, \E^N, \tcE^N  \in \P_2(\Rd)$ that converge to $\R$, $\E$, and $\tcE$ in the 2-Wasserstein metric.  If $\R$ does not give mass to small sets, then for any choices of optimal transport plans $\brho^N \in \Gamma_0(\R^N, \E^N)$ and $\tbrho^N \in \Gamma_0(\R^N,  \tcE ^N)$, we have
\begin{align}
\lim_{N \to +\infty} LOT_{\brho^N,\tbrho^N}(\E^N,\tilde{\E}^N) = W_{2, \R}(\E, \tilde{\E}).
\end{align}
 \end{prop}

 \begin{proof}
Throughout, we   use the equivalence between  convergence in the Wasserstein metric and  narrow convergence combined with convergence of second moments \cite[Remark 7.1.11]{AGS}. In particular, this fact ensures that $\R^N$, $\E^N$, and $\tcE^N$ converge narrowly, so $\bomega^N$ is narrowly relatively compact \cite[Lemma 5.2.2]{AGS}. Any narrow limit point $\bomega$ of this sequence  satisfies, in the sense of narrow convergence,
\begin{align} \label{bomegaunique1}
\pi^{1,2} \#\bomega = \lim_{N \to +\infty}  \pi^{1,2} \# \bomega^{N} =\lim_{N \to +\infty} \brho^N =\brho   ,   \\
\pi^{1,3} \#\bomega=  \lim_{N \to +\infty}  \pi^{1,3} \# \bomega^{N} =\lim_{N \to +\infty} \tbrho^N = \tbrho     , \label{bomegaunique2}
\end{align}
where $\brho \in \Gamma_0(\R, \E)$, $\tbrho \in \Gamma_0(\R, \tcE)$ \cite[Proposition 7.1.3]{AGS}. Since $\R$ doesn't give mass to sets of $(d-1)$-dimensional Hausdorff measure, the limit point $\bomega$ is unique and $\bomega = (\id \times \bt_\R^\E \times \bt_\R^\tcE) \# \R$ \cite[Lemma 5.3.2]{AGS}.
Furthermore, since
\begin{align}\nonumber
&\lim_{N\to +\infty} M_2(\bomega^N) \\ \nonumber
&\quad = \lim_{N \to +\infty} \int |x_1|^2 + |x_2|^2 +|x_3|^2 d \bomega^N(x_1,x_2,x_3) \\ \nonumber
&\quad = \lim_{N \to +\infty} M_2(\R^N) + M_2(\E^N) + M_2(\tcE^N) \\
&\quad = M_2(\R) + M_2(\E) + M_2(\tcE) = M_2(\bomega) ,
\end{align}
we obtain that $\bomega^N \to \bomega$ not only narrowly, but also in the Wasserstein metric.

 We now apply this convergence of $\bomega^N$ to $\bomega$ to conclude the convergence of the LOT approximation to the transport metric with base $\R$. First, we will show
\begin{align}  \label{limsupineq}
\limsup_{N \to +\infty} LOT_{\brho^N,\tbrho^N}(\E^N,\tilde{\E}^N) \leq W_{2, \R}(\E, \tilde{\E}) .
\end{align}
By Jensen's inequality for the probability measures $\bomega^N_{x_1}$,
\begin{align} \nonumber 
&LOT_{\brho^N,\tbrho^N}(\E^N,\tilde{\E}^N) \\ \nonumber
&\quad \leq \left( \iint |x_2 - x_3|^2 d \bomega^N_{x_1}(x_2,x_3)  d \R^N(x_1) \right)^{1/2} \\
&\quad = \left( \int    |x_2 - x_3|^2 d \bomega^N(x_1,x_2,x_3) \right)^{1/2} .
\end{align}
Taking the limsup as $N\to +\infty$ and using the convergence of  $\bomega^N$ to $\bomega = (\id \times \bt_\R^\E \times \bt_\R^\tcE) \# \R$ in the Wasserstein metric gives inequality (\ref{limsupineq}) \cite[Lemma 5.1.7, Proposition 7.1.5]{AGS}.

It remains to show that
\begin{align}  \label{liminfineq}
\liminf_{N \to +\infty} LOT_{\brho^N,\tbrho^N}(\E^N,\tilde{\E}^N) \geq W_{2, \R}(\E, \tilde{\E}) .
\end{align}
Since $\R$ does not give mass to sets of $(d-1)$-dimensional Hausdorff measure, $W_{2, \R}(\E, \tcE) = LOT_{\brho, \brho'}(\E, \tcE) $, and, squaring both sides,  it is equivalent to show
\begin{align}
&\liminf_{N \to +\infty}    \int \left| v^N(x_1) \right|^2 d \R^N(x_1)   \geq   \int \left| v(x_1) \right|^2 d \R(x_1)  ,
\end{align}
where
\begin{align} \nonumber
v^N(x_1)&= \int(x_2 - x_3) d \bomega^N_{x_1}(x_2,x_3) \\
v(x_1)&= \int(x_2 - x_3) d \bomega_{x_1}(x_2,x_3)
\end{align}
Since $\R^N \to \R$ narrowly and $x \mapsto |x|^2$ is convex, this holds as long as $v^N \in L^2(\R^N)$ weakly converge to $v \in L^2(\R)$ \cite[Theorem 5.4.4 (ii)]{AGS}. Indeed, for any $f \in C^\infty_c(\Rd)$, the fact that $\bomega^N \to \bomega$ in the Wasserstein metric ensures
\begin{align} \nonumber
& \lim_{N \to +\infty}   \int f(x_1) v^N(x_1) d \R^N(x_1) \\ \nonumber
 &\quad = \lim_{N \to +\infty} \iint f(x_1) (x_2-x_3) d \bomega^N_{x_1}(x_2,x_3) d \R^N(x_1) \\ \nonumber
 &\quad = \lim_{N \to +\infty} \int f(x_1) (x_2-x_3) d \bomega^N(x_1,x_2,x_3) \\ \nonumber
 &\quad =    \int f(x_1) (x_2-x_3) d \bomega(x_1,x_2,x_3) \\ \nonumber
&\quad = \iint f(x_1) (x_2-x_3) d \bomega_{x_1}(x_2,x_3) d \R(x_1) \\
&\quad =   \int f(x_1) v(x_1) d \R(x_1) .
\end{align}
 \end{proof}
 
 \begin{cor} \label{LOTcor}
Let $\Omega$ be a two dimensional rectangular domain, and consider a sequence of reference measures $\R^N$  given by a   sum of $N^2$ Dirac masses with weights $1/N^2$,  uniformly distributed on    $\Omega$. Then, as $N \to +\infty$, the LOT approximation with base $\R^N$ converges to the transport metric with base $\R$, where $\R$ is the probability measure uniformly distributed on $\Omega$. That is,  for any events $\E, \tcE$, and for any $\brho \in \Gamma_0(\R^N, \E)$, $\tbrho \in \Gamma_0(\R^N, \tcE)$, we have
\begin{align}
&\lim_{N \to _\infty} LOT_{\brho^N, \tbrho^N}(\E, \tcE) = W_{2, \R} (\E, \tcE) .
\end{align} \end{cor}

\begin{proof}
Note that, by construction, $\R^N$ converges in the Wasserstein metric  to the probability measure uniformly distributed on $\Omega$, which does not give mass to small sets. Consequently, the result follows from Proposition \ref{convergenceLOTprop}.
\end{proof}

\bibliography{emdbib}

\end{document}